\documentclass[envcountsame,english]{llncs}

\usepackage{babel}
\usepackage{latexsym,amssymb,amsmath,ifthen,alltt,array,enumerate}
\usepackage{url,color}
\usepackage{graphicx}
\usepackage{xspace}
\usepackage{extarrows}
\RequirePackage{etex}
\usepackage{tikz}
\usetikzlibrary{decorations,arrows}
\RequirePackage{pifont}

\newcommand{\TTT}{\textsf{T\!\raisebox{-1mm}{T}\!T}\xspace}
\newcommand{\TTTT}{$\TTT\textsf{\!\raisebox{-1mm}{2}}$\xspace}

\newcommand{\m}[1]{\mathsf{#1}}

\renewcommand{\AA}{\mathcal{A}}

\newcommand{\CC}{\mathcal{C}}

\newcommand{\FF}{\mathcal{F}}

\newcommand{\MM}{\mathcal{M}}

\newcommand{\PP}{\mathcal{P}}
\newcommand{\PR}{\mathsf{CPS}}

\newcommand{\RL}{{\m{RL}}}
\newcommand{\RR}{\mathcal{R}}

\renewcommand{\SS}{\mathcal{S}}

\newcommand{\VV}{\mathcal{V}}

\newcommand{\TERM}{\mathcal{T(F,V)}}

\newcommand{\Var}{\VV\m{ar}}

\newcommand{\Pos}{\PP\m{os}}
\newcommand{\FPos}{\Pos_\FF}
\newcommand{\VPos}{\Pos_\VV}

\newcommand{\seq}[2][n]{{#2_1},\dots,{#2_{#1}}}

\newcommand{\nd}{\mathsf{nd}}
\newcommand{\du}{\mathsf{d}}

\newcommand{\from}{\leftarrow}

\newcommand{\FromB}[1]{\mathrel{{\vphantom{\to}_{#1}}{\from}}}
\newcommand{\cp}{\mathrel{\from\!\rtimes\!\to}}
\newcommand{\cpab}{\mathrel{\FromB{\alpha}\!\rtimes\!\to_\beta}}

\def\test#1#2#3{\setbox0=\hbox{$\vphantom{#1}^{#2}_{#3}$}%
                \dimen0=\wd0%
                \setbox1=\hbox{$\scriptstyle #2$}%
                \advance\dimen0-\wd1%
                \setbox1=\hbox{\hskip\dimen0\copy1}%
                \dimen0=\wd0%
                \setbox2=\hbox{$\scriptstyle #3$}%
                \advance\dimen0-\wd2%
                \setbox2=\hbox{\hskip\dimen0\copy2}%
                {\vphantom{#1}^{\box1}_{\box2}}{#1}
}
\newcommand{\From}[2]{\mathrel{\test{\from}{#2}{#1}}}

\newcommand{\ARS}{\langle A, \to \rangle}
\newcommand{\ARSn}{\langle A, \{ \to_\alpha \}_{\alpha \in I} \rangle}
\newcommand{\dD}[1]{\smash{\raisebox{-.5mm}{\rotatebox{#1}{%
$\scriptstyle <$}}}}
\newcommand{\dd}{\dD{90}}
\newcommand{\ddd}{\smash{\raisebox{-.5mm}{\rotatebox{45}{%
$\scriptstyle <$}}}}
\newcommand{\dDd}{\smash{\raisebox{0mm}{\rotatebox{135}{%
$\scriptstyle <$}}}}

\newcommand{\eee}{\smash{\raisebox{0mm}{\rotatebox{-45}{%
$\scriptstyle =$}}}}
\newcommand{\eEe}{\smash{\raisebox{-1mm}{\rotatebox{45}{%
$\scriptstyle =$}}}}
\newcommand{\LD}{\text{LD}}
\newcommand{\dDg}[1]{\smash{\raisebox{-.5mm}{\rotatebox{#1}{%
$\scriptstyle \eqslantless$}}}}
\newcommand{\ddg}{\dDg{90}}
\newcommand{\lda}[1]{\xleftarrow{\smash{%
\raisebox{-2mm}{\makebox[4mm]{$\stackrel{#1}{\circ}$}}}}}
\newcommand{\rda}[1]{\xrightarrow{\smash{%
\raisebox{-2mm}{\makebox[4mm]{$\stackrel{#1}{\circ}$}}}}}

\begin{document}

\bibliographystyle{lncs}
\pagestyle{plain}

\title{Decreasing Diagrams and Relative Termination}

\author{Nao Hirokawa\inst{1} \and Aart Middeldorp\inst{2}}
\institute{
School of Information Science \\
Japan Advanced Institute of Science and Technology, Japan \\
\email{hirokawa@jaist.ac.jp} \\[2ex]
\and
Institute of Computer Science \\
University of Innsbruck, Austria \\
\email{aart.middeldorp@uibk.ac.at}
}

\maketitle

\begin{abstract}
In this paper we use the decreasing diagrams technique to show that a
left-linear term rewrite system $\RR$ is confluent if all its critical
pairs are joinable and the critical pair steps are relatively terminating
with respect to $\RR$. We further show how to encode the rule-labeling
heuristic for decreasing diagrams as a satisfiability problem. Experimental
data for both methods are presented.
\end{abstract}

\section{Introduction}

This paper is concerned with automatically proving confluence of
term rewrite systems. Unlike termination, for which the interest in
automation gave and continues to give rise to new methods and tools,%
\footnote{\url{http://termination-portal.org/wiki/Termination_Competition}}
automating confluence has received little attention. Only very recently,
the first confluence tool made its appearance: ACP~\cite{ACP} implements
Knuth and Bendix' condition---joinability of critical pairs---for
terminating rewrite systems~\cite{KB70}, several critical pair criteria
for left-linear rewrite systems~\cite{H80,T88,vO97}, as well as
divide and conquer techniques based on persistence~\cite{AT97},
layer-preservation~\cite{O94}, and commutativity~\cite{R73}.

For abstract rewrite systems, the \emph{decreasing diagrams} technique of
van Oostrom~\cite{vO94} subsumes all sufficient conditions for confluence.
To use this technique for term rewrite systems, a well-founded order on
the rewrite steps has to be supplied such that rewrite peaks can be
completed into so-called decreasing diagrams.

We present two results in this paper. We show how to encode the
rule-labeling heuristic of van Oostrom~\cite{vO08} for linear rewrite
systems as a satisfiability problem. In this heuristic rewrite steps are
labeled by the applied rewrite rule. By limiting the number of steps that
may be used to complete local diagrams, we obtain a finite search problem
which is readily transformed into a satisfiability problem.
Any satisfying assignment returned by a modern SAT or SMT
solver is then translated back into a concrete rule-labeling.

Our second and main result employs the decreasing diagrams technique to
obtain a
new confluence result for left-linear but not necessarily right-linear
rewrite systems. It requires that the rewrite steps involved in the
generation of critical pairs are \emph{relatively terminating} with
respect to the rewrite system. This result can be viewed as a 
generalization of the two standard approaches for proving confluence:
orthogonality and joinability of critical pairs for terminating systems.
In the non-trivial proof we use the self-labeling heuristic in which
rewrite steps are labeled by their starting term.
 
In the next section we recall the decreasing diagrams technique
and present a small variation which better serves our purposes.
Section~\ref{confluence via relative termination} is devoted to our
main result. We prove that a locally confluent left-linear term rewrite
system is confluent if there are no infinite rewrite sequences that
involve infinitely many steps that were used in the generation of critical
pairs.
In Section~\ref{automation} we explain how this result is implemented.
Moreover, we show how the rule-labeling heuristic for decreasing diagrams
can be transformed into a satisfiability problem.
Section~\ref{experiments} presents experimental data.
In Section~\ref{conclusion} we
conclude with suggestions for future research.

\section{Decreasing Diagrams}
\label{decreasing diagrams}

We start this preliminary section by recalling the decreasing diagrams
technique for abstract rewrite systems (ARSs).
We write $\ARSn$ to denote the ARS
$\ARS$ where $\to$ is the union of $\to_\alpha$ for all $\alpha \in I$.
If $J \subseteq I$ then $\to_J$ denotes the union of
$\to_\alpha$ for all $\alpha \in J$. In order to
reduce the number of arrows, we denote the individual relations
$\to_\alpha$ of an ARS $\ARSn$ simply by $\alpha$.

Let $\AA  = \ARSn$ be an ARS and let $>$ be a well-founded order on $I$.
For every $\alpha \in I$ we write $\xrightarrow{\dd}_\alpha$ for the union
of $\to_\beta$ for all $\beta < \alpha$. If $\alpha, \beta \in I$ then
$\xrightarrow{\dd}_{\alpha\beta}$ denotes the union of
$\xrightarrow{\dd}_\alpha$ and $\xrightarrow{\dd}_\beta$.
We say that $\alpha$ and $\beta$ are \emph{locally decreasing} with
respect to $>$ and we write $\LD_>(\alpha,\beta)$ if
\[
{\FromB{\alpha} \cdot \to_\beta} \:\subseteq\:
{\smash{\xrightarrow{\dd}}_\alpha^*
\cdot \to_\beta^=
\cdot \mathrel{\smash{\xrightarrow{\dd}}_{\alpha\beta}^*}
\cdot \mathrel{\test{\xleftarrow{\dd}}{*}{\alpha\beta}}
\cdot \From{\alpha}{=}
\cdot \mathrel{\test{\xleftarrow{\dd}}{*}{\beta}}}
\]
Graphically (dashed arrows are implicitly existentially quantified and
double-headed arrows denote reflexive and transitive closure):
\[
\tikzstyle{element}=[circle]
\tikzstyle{ll}=[text width=1pt,text height=0pt]
\begin{tikzpicture}[scale=.7]
\node at (3,6) [element] (a) {$\cdot$};
\node at (0,3) [element] (b) {$\cdot$};
\node at (6,3) [element] (c) {$\cdot$};
\draw [->] (a) to (b);
\draw [->] (a) to (c);
\node[ll,anchor=south east] at (1.5,4.5)
{\makebox[0mm]{$\scriptstyle \alpha$}};
\node[ll,anchor=south west] at (4.5,4.5)
{\makebox[0mm]{$\scriptstyle \beta$}};
\node at (1,2) [element] (b1) {$\cdot$};
\node at (2,1) [element] (b2) {$\cdot$};
\node at (3,0) [element] (d)  {$\cdot$};
\node at (4,1) [element] (c2) {$\cdot$};
\node at (5,2) [element] (c1) {$\cdot$};
\draw [dash pattern=on 2pt off 2pt,->>]  (b) to (b1);
\draw [dash pattern=on 2pt off 2pt,->]  (b1) to (b2);
\draw [dash pattern=on 2pt off 2pt,->>] (b2) to (d);
\draw [dash pattern=on 2pt off 2pt,->>]  (c) to (c1);
\draw [dash pattern=on 2pt off 2pt,->]  (c1) to (c2);
\draw [dash pattern=on 2pt off 2pt,->>] (c2) to (d);
\node[ll,anchor=north east] at (0.5,2.5)
{\makebox[0mm]{$\scriptstyle \alpha$}};
\node[ll,anchor=south west] at (0.5,2.5)
{\makebox[0mm]{$\ddd$\hspace{1mm}}};
\node[ll,anchor=north east] at (1.5,1.5)
{\makebox[0mm]{$\scriptstyle \beta$}};
\node[ll,anchor=south west] at (1.5,1.5)
{\makebox[0mm]{$\eee$}};
\node[ll,anchor=north east] at (2.5,0.5)
{\makebox[0mm]{$\scriptstyle \alpha\beta$}};
\node[ll,anchor=south west] at (2.5,0.5)
{\makebox[0mm]{$\ddd$\hspace{1mm}}};
\node[ll,anchor=north west] at (5.5,2.5)
{\makebox[0mm]{$\scriptstyle \beta$}};
\node[ll,anchor=south east] at (5.5,2.5)
{\makebox[0mm]{$\hspace{2mm}\dDd$}};
\node[ll,anchor=north west] at (4.5,1.5)
{\makebox[0mm]{$\scriptstyle \alpha$}};
\node[ll,anchor=south east] at (4.5,1.5)
{\makebox[0mm]{$\hspace{2mm}\eEe$}};
\node[ll,anchor=north west] at (3.5,0.5)
{\makebox[0mm]{$\scriptstyle \alpha\beta$}};
\node[ll,anchor=south east] at (3.5,0.5)
{\makebox[0mm]{$\hspace{2mm}\dDd$}};
\node at (3,3) {$\LD_>(\alpha,\beta)$};
\end{tikzpicture}
\]
The ARS $\AA = \ARSn$ is locally decreasing if there exists a well-founded
order $>$ on $I$ such that $\LD(\alpha,\beta)$ for all
$\alpha, \beta \in I$.

Van Oostrom~\cite{vO94} obtained the following result.

\begin{theorem}
\label{dd}
Every locally decreasing ARS is confluent.
\qed
\end{theorem}

Variations of this fundamental confluence result are presented in
\cite{BKO98,KOV00,vO08}.
We present a version of Theorem~\ref{dd} which
is more suitable for our purposes.

Let $\AA  = \ARSn$ be an ARS. Let $(>,\geqslant)$ consist of a well-founded
order $>$ on $I$ together with a quasi-order $\geqslant$ such that
${\geqslant} \cdot {>} \cdot {\geqslant} \subseteq {>}$.
For every $\alpha \in I$ we write $\xrightarrow{\ddg}_\alpha$ for the
union of $\to_\beta$ for all $\beta \leqslant \alpha$ and all
$\beta < \alpha$. (Note that ${>} \subseteq {\geqslant}$
need not hold.)
We say that $\alpha$ and $\beta$ are locally decreasing with respect to
$(>,\geqslant)$ and we write $\LD_{(>,\geqslant)}(\alpha,\beta)$ if
\[
{\FromB{\alpha} \cdot \to_\beta} \:\subseteq\:
{\smash{\xrightarrow{\dd}}_\alpha^*
\cdot \xrightarrow{\ddg}_\beta^{=}
\cdot \mathrel{\smash{\xrightarrow{\dd}}_{\alpha\beta}^*}
\cdot \mathrel{\test{\xleftarrow{\dd}}{*}{\alpha\beta}}
\cdot \mathrel{\test{\xleftarrow{\ddg}}{=}{\alpha}}
\cdot \mathrel{\test{\xleftarrow{\dd}}{*}{\beta}}}
\]
The ARS $\AA = \ARSn$ is \emph{extended} locally decreasing if there
exists $(>,\geqslant)$ such that $\LD_{(>,\geqslant)}(\alpha,\beta)$
for all $\alpha, \beta \in I$.

\begin{theorem}
\label{dde}
Every extended locally decreasing ARS is confluent.
\end{theorem}
\begin{proof}
Let $\AA = \ARSn$ be extended locally decreasing with respect to
$(>,\geqslant)$.
We write $C_\alpha$ as the set of all $\beta \in I$ with
$\alpha \geqslant \beta$ and $\alpha \not> \beta$.
The set of all such $C_\alpha$ is denoted by $\CC$.
For every $C \in \CC$ we write $\to_C$ for the union of
$\to_\alpha$ for all $\alpha \in C$. The well-founded order $>$ on $I$
can be lifted to $\CC$: $C_\alpha > C_\beta$ if $\alpha > \beta$.
If $C_\alpha = C_{\alpha'}$,
$C_\beta = C_{\beta'}$, and $\alpha > \beta$ then
$\alpha' \geqslant \alpha > \beta \geqslant \beta'$ and thus
$\alpha' > \beta'$
because of the requirement
${\geqslant} \cdot {>} \cdot {\geqslant} \subseteq {>}$.
Hence $>$ is well-defined on $\CC$.
If $\beta \leqslant \alpha$ and $\beta < \alpha$ then
${\to_\beta} \subseteq {\to_{C_\beta}} \subseteq
{\xrightarrow{\dd}_{C_\alpha}}$.
If $\beta \leqslant \alpha$ and $\beta \not< \alpha$ then
$\beta \in C_\alpha$ and thus
${\to_\beta} \subseteq {\to_{C_\alpha}}$.
Hence ${\xrightarrow{\ddg}_\alpha} \subseteq
{\xrightarrow{\dd}_{C_\alpha}} \cup {\to_{C_\alpha}}$.
Now consider arbitrary sets $C, D \in \CC$ and let $\alpha \in C$ and
$\beta \in D$. From the assumption $\LD_{(>,\geqslant)}(\alpha,\beta)$ we
obtain
\[
{\FromB{\alpha} \cdot \to_\beta} \:\subseteq\:
{\smash{\xrightarrow{\dd}}_\alpha^*
\cdot \xrightarrow{\ddg}_\beta^=
\cdot \mathrel{\smash{\xrightarrow{\dd}}_{\alpha\beta}^*}
\cdot \mathrel{\test{\xleftarrow{\dd}}{*}{\alpha\beta}}
\cdot \mathrel{\test{\xleftarrow{\ddg}}{=}{\alpha}}
\cdot \mathrel{\test{\xleftarrow{\dd}}{*}{\beta}}}
\]
By construction, the latter relation is contained in
\[
{\smash{\xrightarrow{\dd}}_C^*
\cdot \xrightarrow{\dd}_D^=
\cdot \mathrel{\smash{\xrightarrow{\dd}}_{CD}^*}
\cdot \mathrel{\test{\xleftarrow{\dd}}{*}{CD}}
\cdot \mathrel{\test{\xleftarrow{\dd}}{=}{C}}
\cdot \mathrel{\test{\xleftarrow{\dd}}{*}{D}}}
\]
Since
\[
{\FromB{C} \cdot \to_D} \:=\:
\makebox[7mm]{$\displaystyle \bigcup_{\alpha \in C,\,\beta \in D}$}
{\FromB{\alpha} \cdot \to_\beta}
\]
we conclude $\LD_>(C,D)$. According to Theorem~\ref{dd}, the ARS
$\langle A, \{ \to_C \}_{C \in \CC} \rangle$ is confluent. Since
\[
\bigcup_{\alpha \in I} \to_\alpha \:=\:
\makebox[5mm]{$\displaystyle \bigcup_{C \in \CC}$} \to_C
\]
it follows that $\AA$ is confluent.
\qed
\end{proof}

We are interested in the application of Theorems~\ref{dd} and~\ref{dde}
for proving confluence of term rewrite systems (TRSs).

Many sufficient conditions for confluence of TRSs are based on critical
pairs. Critical pairs are generated from overlaps. An overlap
$(l_1 \to r_1, p, l_2 \to r_2)_\mu$ of a TRS $\RR$ consists of variants
$l_1 \to r_1$ and $l_2 \to r_2$ of rewrite rules of $\RR$
without common variables, a position $p \in \Pos_\FF(l_2)$, and
a most general unifier $\mu$ of $l_1$ and $l_2|_p$.
If $p = \epsilon$ then we require that $l_1 \to r_1$ and $l_2 \to r_2$ are
not variants. The induced critical pair is
$l_2\mu[r_1\mu]_p \approx r_2\mu$.
Following Dershowitz~\cite{D05}, we write $s \cp t$ to indicate that
$s \approx t$ is a critical pair.

In \cite{vO08} van Oostrom proposed the \emph{rule-labeling heuristic}
in which rewrite steps are partitioned according to the employed rewrite
rules. If one can find an order on the rules of a \emph{linear} TRS such
that every critical pair is locally decreasing, confluence is guaranteed.
A formalization of this heuristic is given below where ${\cpab}$
denotes the set of critical pairs obtained from overlaps
$(\alpha, p, \beta)_\mu$.

\begin{theorem}
\label{dd-ll}
A linear TRS $\RR$ is confluent if there exists a well-founded order
$>$ on the rules of $\RR$ such that
\[
{\cpab} \subseteq 
{\smash{\xrightarrow{\dd}}_\alpha^*
\cdot \xrightarrow{\ddg}_\beta^{=}
\cdot \mathrel{\smash{\xrightarrow{\dd}}_{\alpha\beta}^*}
\cdot \mathrel{\test{\xleftarrow{\dd}}{*}{\alpha\beta}}
\cdot \mathrel{\test{\xleftarrow{\ddg}}{=}{\alpha}}
\cdot \mathrel{\test{\xleftarrow{\dd}}{*}{\beta}}}
\]
for all rewrite rules $\alpha, \beta \in \RR$.
Here $\geqslant$ is the reflexive closure of $>$.
\qed
\end{theorem}

The heuristic readily applies to the following example from \cite{GL06}.

\begin{example}
\label{GL}
Consider the linear TRS $\RR$ consisting of the rewrite rules
\begin{alignat*}{4}
1\colon && \m{nat} &\to \m{0} : \m{inc}(\m{nat}) & \qquad
4\colon &\;& \m{inc}(x : y) &\to \m{s}(x) : \m{inc}(y) \\
2\colon && \m{hd}(x : y) &\to x &
5\colon &\;& \m{inc}(\m{tl}(\m{nat})) &\to \m{tl}(\m{inc}(\m{nat})) \\
3\colon && \m{tl}(x : y)  &\to y &
\end{alignat*}
There is one critical pair:
\[
\tikzstyle{element}=[rectangle]
\tikzstyle{ll}=[text width=1pt,text height=0pt]
\begin{tikzpicture}
\node at (1.5,1) [element] (a) {$\m{inc}(\m{tl}(\m{nat}))$};
\node at (0,0) [element] (b) {$\m{inc}(\m{tl}(\m{0} : \m{inc}(\m{nat})))$};
\node at (3,0) [element] (c) {$\m{tl}(\m{inc}(\m{nat}))$};
\draw [->] (a) to (b);
\node[ll,anchor=south east] at (0.75,0.5)
{\makebox[0mm]{$\scriptstyle 1$}};
\draw [->] (a) to (c);
\node[ll,anchor=south west] at (2.25,0.5)
{\makebox[0mm]{$\scriptstyle 5$}};
\end{tikzpicture}
\]
We have
\begin{align*}
\m{inc}(\m{tl}(\m{0} : \m{inc}(\m{nat}))) 
\xrightarrow{3} {} & \m{inc}(\m{inc}(\m{nat})) \\
\m{tl}(\m{inc}(\m{nat}))
\xrightarrow{1} \m{tl}(\m{inc}(\m{0} : \m{inc}(\m{nat})))
\xrightarrow{4} \m{tl}(\m{s}(\m{0}) : \m{inc}(\m{inc}(\m{nat}))) 
\xrightarrow{3} {} & \m{inc}(\m{inc}(\m{nat}))
\end{align*}
Hence the critical pair is locally decreasing with respect to the
rule-labeling heuristic together with the order $5 > 1, 3, 4$.
\end{example}

The following example (Vincent van Oostrom, personal communication)
shows that linearity in Theorem~\ref{dd-ll} cannot be weakened to
left-linearity.

\begin{example}
\label{counterexample1}
Consider the TRS $\RR$ consisting of the rewrite rules
\begin{xalignat*}{4}
1\colon \m{f}(\m{a},\m{a}) &\to \m{c} &
2\colon \m{f}(\m{b},x) &\to \m{f}(x,x) &
3\colon \m{f}(x,\m{b}) &\to \m{f}(x,x) &
4\colon \m{a} &\to \m{b}
\end{xalignat*}
There are three critical pairs:
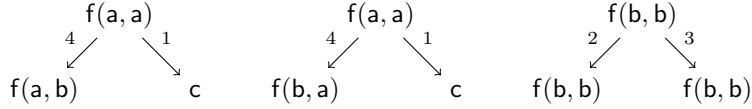
\[
\tikzstyle{element}=[rectangle]
\tikzstyle{ll}=[text width=1pt,text height=0pt]
\begin{tikzpicture}
\node at (1,1) [element] (a) {$\m{f}(\m{a},\m{a})$};
\node at (0,0) [element] (b) {$\m{f}(\m{a},\m{b})$};
\node at (2,0) [element] (c) {$\m{c}\mathstrut$};
\draw [->] (a) to (b);
\node[ll,anchor=south east] at (0.5,0.5)
{\makebox[0mm]{$\scriptstyle 4$}};
\draw [->] (a) to (c);
\node[ll,anchor=south west] at (1.5,0.5)
{\makebox[0mm]{$\scriptstyle 1$}};
\end{tikzpicture}
\qquad
\begin{tikzpicture}
\node at (1,1) [element] (a) {$\m{f}(\m{a},\m{a})$};
\node at (0,0) [element] (b) {$\m{f}(\m{b},\m{a})$};
\node at (2,0) [element] (c) {$\m{c}\mathstrut$};
\draw [->] (a) to (b);
\node[ll,anchor=south east] at (0.5,0.5)
{\makebox[0mm]{$\scriptstyle 4$}};
\draw [->] (a) to (c);
\node[ll,anchor=south west] at (1.5,0.5)
{\makebox[0mm]{$\scriptstyle 1$}};
\end{tikzpicture}
\qquad
\begin{tikzpicture}
\node at (1,1) [element] (a) {$\m{f}(\m{b},\m{b})$};
\node at (0,0) [element] (b) {$\m{f}(\m{b},\m{b})$};
\node at (2,0) [element] (c) {$\m{f}(\m{b},\m{b})$};
\draw [->] (a) to (b);
\node[ll,anchor=south east] at (0.5,0.5)
{\makebox[0mm]{$\scriptstyle 2$}};
\draw [->] (a) to (c);
\node[ll,anchor=south west] at (1.5,0.5)
{\makebox[0mm]{$\scriptstyle 3$}};
\end{tikzpicture}
\]
We have $\m{f}(\m{a},\m{b})
\xrightarrow{3} \m{f}(\m{a},\m{a})
\xrightarrow{1} \m{c}$
and $\m{f}(\m{b},\m{a})
\xrightarrow{2} \m{f}(\m{a},\m{a})
\xrightarrow{1} \m{c}$. It follows that the critical pairs are locally
decreasing by taking the order $4 > 2,3$. Nevertheless, the conversion
$\m{f}(\m{b},\m{b}) \from \m{f}(\m{b},\m{a}) \from \m{f}(\m{a},\m{a})
\to \m{c}$ reveals that $\RR$ is not confluent.
\end{example}

In the next section we show that the restriction to right-linear rewrite
rules can be dropped by imposing a relative termination condition. 

\section{Confluence via Relative Termination}
\label{confluence via relative termination}

In this section we weaken the right-linearity requirement in
Theorem~\ref{dd-ll}.
Let $\RR$ be a TRS. We denote the set
\[
\{ 
l_2\mu \to l_2\mu[r_1\mu]_p, l_2\mu \to r_2\mu \mid 
\text{$(l_1 \to r_1, p, l_2 \to r_2)_\mu$ is an overlap of $\RR$} \}
\]
of rewrite steps that give rise to critical pairs of $\RR$ by
$\PR(\RR)$.
The rules in $\PR(\RR)$ are called \emph{critical pair steps}.
We say that $\RR$ is relatively terminating with respect to $\SS$ if the
relation $\to_\SS^* \cdot \to_\RR^{} \cdot \to_\SS^*$ is well-founded.
The main result of this section (Theorem~\ref{dd-l2} below) states that
a left-linear locally confluent TRS $\RR$ is confluent if 
$\PR(\RR)$ is relatively terminating with respect to $\RR$. In the
proof we use decreasing diagrams with the \emph{self-labeling heuristic}
in which rewrite steps are labeled by their starting term. A key
problem when trying to prove confluence in the absence of termination is
the handling of duplicating rules.
Parallel rewrite steps are typically used for this purpose. To anticipate
future developments (cf.\ Section~\ref{conclusion}) we use
complete development steps instead.
However, first we present a special case of our main result in which
duplicating rules are taken care of by requiring them to be relatively
terminating with respect to the non-duplicating ones.

\begin{theorem}
\label{dd-l1}
Let $\RR$ be a left-linear TRS. Let $\RR_\du$ be the subset of duplicating
rules and $\RR_\nd$ the subset of non-duplicating rules in $\RR$. The TRS
$\RR$ is confluent if ${\cp} \subseteq {\downarrow}$ and
$\PR(\RR) \cup \RR_\du$ is relatively terminating with respect to
$\RR_\nd$.
\end{theorem}
\begin{proof}
We label rewrite steps by their starting term. Labels are compared with
respect to the strict order
${>} = {\to_{(\PR(\RR) \cup \RR_\du)/\RR_\nd}^+}$ and
the quasi-order ${\geqslant} = {\to_\RR^*}$.
Note that $>$ is well-founded by the assumption that
$\PR(\RR) \cup \RR_\du$ is relatively terminating with respect to
$\RR_\nd$.
We show that all local peaks of $\RR$ are extended locally decreasing.
Let $s \to t_1$ and $s \to t_2$ by applying the rewrite rules
$l_1 \to r_1$ and $l_2 \to r_2$ at the positions $p_1$ and $p_2$.
We may assume that $l_1 \to r_1$ and $l_2 \to r_2$ do not share variables
and thus there exists a substitution $\sigma$ such that
$s = s[l_1\sigma]_{p_1} = s[l_2\sigma]_{p_2}$, $t_1 = s[r_1\sigma]_{p_1}$,
and $t_2 = s[r_2\sigma]_{p_2}$.
We distinguish three cases.
\begin{enumerate}
\item
If $p_1 \mathop{\|} p_2$ then $t_1 \to u \from t_2$
for the term $u = s[r_1\sigma,r_2\sigma]_{p_1,p_2}$.
We have $s > t_1$ if $l_1 \to r_1$ is duplicating and
$s \geqslant t_1$ if $l_1 \to r_1$ is non-duplicating. So in
both cases we have $t_1 \xrightarrow{\ddg}_s u$. Similarly,
$t_2 \xrightarrow{\ddg}_s u$ and thus we have local decreasingness.
\item
Suppose the redexes $l_1\sigma$ at position $p_1$ and
$l_2\sigma$ at position $p_2$ overlap. If $p_1 = p_2$ and $l_1 \to r_1$ and
$l_2 \to r_2$ are variants then $t_1 = t_2$ and there is nothing to prove.
Assume without loss of generality that $p_1 \leqslant p_2$.
There exists a substitution $\tau$ such that $t_1 = s[v\tau]_{p_1}$ and
$t_2 = s[u\tau]_{p_1}$ with $u \cp v$. By assumption $u \downarrow v$ and
hence also $t_1 \downarrow t_2$. Every label $a$ in the valley between
$t_1$ and $t_2$ satisfies $t_1 \geqslant a$ or $t_2 \geqslant a$. Since
$s \to_{\PR(\RR)} t_1$ and $s \to_{\PR(\RR)} t_2$, it follows that
$s > t_1, t_2$. Hence $s > a$ for every label $a$ in the valley between
$t_1$ and $t_2$. Consequently, local decreasingness holds.
\item
In the remaining case we have a variable overlap. Assume without loss of
generality that $p_1 < p_2$. Let $x$ be the variable in $l_1$ whose
position is above $p_2 \setminus p_1$. Due to linearity of $l_1$ we have
$t_1 \to^* u \from t_2$ for some term $u$. The number of steps in the
sequence from $t_1$ to $u$ equals the number of occurrences of the
variable $x$ in $r_1$. If this number is not more than one then
local decreasingness is obtained as in the first case. If this number is
more than one then $l_1 \to r_1$ is duplicating and hence $s > t_1$. 
Therefore $s > a$ for every term $a$ in the sequence from $t_1$ to $u$.
Moreover $s > t_2$ or $s \geqslant t_2$. Hence also in this
case we have local decreasingness.
\qed
\end{enumerate}
\end{proof}

\begin{example}
Consider the TRS $\RR$ from \cite[p.28]{G96} consisting of the rewrite
rules
\begin{xalignat*}{3}
\m{f}(\m{g}(x))    & \to \m{f}(\m{h}(x,x))   &
\m{g}(\m{a})       & \to \m{g}(\m{g}(\m{a})) &
\m{h}(\m{a},\m{a}) & \to \m{g}(\m{g}(\m{a}))
\end{xalignat*}
The only critical pair 
$\m{f}(\m{g}(\m{g}(\m{a}))) \cp \m{f}(\m{h}(\m{a},\m{a}))$
is clearly joinable. The TRS $\PR(\RR) \cup \RR_\du$ consists of
the rewrite rules
\begin{xalignat*}{3}
\m{f}(\m{g}(\m{a})) & \to \m{f}(\m{h}(\m{a},\m{a})) &
\m{f}(\m{g}(\m{a})) & \to \m{f}(\m{g}(\m{g}(\m{a}))) &
\m{f}(\m{g}(x))     & \to \m{f}(\m{h}(x,x))
\end{xalignat*}
and can be shown to be relatively terminating with respect to
$\RR_\nd$ using the method described at the beginning of
Section~\ref{automation}.
Hence the confluence of $\RR$ is concluded by Theorem~\ref{dd-l1}.
\end{example}

The following example shows that left-linearity is essential in
Theorem~\ref{dd-l1}.

\begin{example}
Consider the non-left-linear TRS $\RR$
\begin{xalignat*}{3}
\m{f}(x,x) & \to \m{a}
& \m{f}(x,\m{g}(x)) & \to \m{b}
& \m{c} & \to \m{g}(\m{c}) 
\end{xalignat*}
from \cite{H80}.
Since $\PR(\RR) \cup \RR_\du$ is empty, termination of 
$(\PR(\RR) \cup \RR_\du)/\RR_\nd$ is trivial.
However, $\RR$ is not confluent because the term $\m{f}(\m{c},\m{c})$
has two distinct normal forms.
\end{example}

The termination of $(\PR(\RR) \cup \RR_\du)/\RR_\nd$ can be weakened to
the termination of $\PR(\RR)/\RR$. Since $\PR(\RR)$ is empty for
every orthogonal TRS $\RR$, we obtain a generalization of orthogonality.

In the proof of Theorem~\ref{dd-l1} we showed the local decreasingness 
of $\to_\RR$. The following example shows that this no longer holds
under the weakened termination assumption.

\begin{example}
Consider the orthogonal TRS $\RR$ consisting of the two rules
$\m{f}(x) \to \m{g}(x,x)$ and $\m{a} \to \m{b}$. Consider the local peak
$\m{f}(\m{b}) \from \m{f}(\m{a}) \to \m{g}(\m{a},\m{a})$.
There are two ways to complete the diagram:
\[
\tikzstyle{element}=[rectangle]
\tikzstyle{ll}=[text width=1pt,text height=0pt]
\begin{tikzpicture}
\node at (0,2) [element] (a) {$\m{f}(\m{a})$};
\node at (0,0) [element] (b) {$\m{f}(\m{b})$};
\node at (2,2) [element] (c) {$\m{g}(\m{a},\m{a})$};
\node at (2,1) [element] (d) {$\m{g}(\m{b},\m{a})$};
\node at (2,0) [element] (e) {$\m{g}(\m{b},\m{b})$};
\draw [->] (a) to (b) node[ll,pos=0.5,anchor=east]
  {\makebox[0mm][r]{$\scriptstyle \m{f}(\m{a})$}};
\draw [->] (a) to (c) node[ll,pos=0.5,anchor=south]
  {\makebox[0mm]{$\scriptstyle \m{f}(\m{a})$}};
\draw [->] (c) to (d) node[ll,pos=0.5,anchor=west]
  {$\scriptstyle \m{g}(\m{a},\m{a})$};
\draw [->] (d) to (e) node[ll,pos=0.5,anchor=west]
  {$\scriptstyle \m{g}(\m{b},\m{a})$};
\draw [->] (b) to (e) node[ll,pos=0.5,anchor=north]
  {\makebox[0mm]{$\scriptstyle \m{f}(\m{b})$}};
\end{tikzpicture}
\qquad\qquad
\begin{tikzpicture}
\node at (0,2) [element] (a) {$\m{f}(\m{a})$};
\node at (0,0) [element] (b) {$\m{f}(\m{b})$};
\node at (2,2) [element] (c) {$\m{g}(\m{a},\m{a})$};
\node at (2,1) [element] (d) {$\m{g}(\m{a},\m{b})$};
\node at (2,0) [element] (e) {$\m{g}(\m{b},\m{b})$};
\draw [->] (a) to (b) node[ll,pos=0.5,anchor=east]
  {\makebox[0mm][r]{$\scriptstyle \m{f}(\m{a})$}};
\draw [->] (a) to (c) node[ll,pos=0.5,anchor=south]
  {\makebox[0mm]{$\scriptstyle \m{f}(\m{a})$}};
\draw [->] (c) to (d) node[ll,pos=0.5,anchor=west]
  {$\scriptstyle \m{g}(\m{a},\m{a})$};
\draw [->] (d) to (e) node[ll,pos=0.5,anchor=west]
  {$\scriptstyle \m{g}(\m{a},\m{b})$};
\draw [->] (b) to (e) node[ll,pos=0.5,anchor=north]
  {\makebox[0mm]{$\scriptstyle \m{f}(\m{b})$}};
\end{tikzpicture}
\]
Since $\PR(\RR)$ is empty, neither of them is
extended locally decreasing
in the sense of the proof of Theorem~\ref{dd-l1}.
\end{example}

To address the problem, we first recall \emph{complete developments}.

\begin{definition}
Let $\RR$ be a TRS. The relation $\rda{}_\RR$ is inductively defined as
follows:
\begin{enumerate}[(1)]
\item
$x \rda{}_\RR x$ for all variables $x$,
\item
$f(\seq{s}) \rda{}_\RR f(\seq{t})$ if 
for each $i$ we have $s_i \rda{}_\RR t_i$, and
\item
$l\sigma \rda{}_\RR r\tau$ if $l \to r \in \RR$ and
$\sigma \rda{}_\RR \tau$.
\end{enumerate}
where $\sigma \rda{}_\RR \tau$ if $x\sigma \rda{}_\RR x\tau$
for all variables $x$.
\end{definition}

\begin{lemma}
\label{inclusion}
For every TRS $\RR$ we have
${\to_\RR} \subseteq {\rda{}_\RR} \subseteq {\to_\RR^*}$.
\qed
\end{lemma}

The following lemma relates $\rda{}_\RR$ to $\to_{\PR(\RR)/\RR}$. It is
the key to prove our main result.

\begin{lemma}
\label{development}
Let $\RR$ be a TRS and $l \to r$ a left-linear rule in $\RR$.
If $l\sigma \rda{}_\RR t$ then one of the following conditions holds:
\begin{enumerate}[(a)]
\item
$t \in \{ l\tau, r\tau \}$ and $\sigma \rda{}_\RR \tau$ for some $\tau$,
\item
$l\sigma \to_{\PR(\RR)} \cdot \rda{}_{\RR} t$
and $l\sigma \to_{\PR(\RR)} r\sigma$.
\end{enumerate}
\end{lemma}
\begin{proof}
We may write $l\sigma = C[\seq{s}] \rda{}_\RR C[\seq{t}] = t$ where
$s_i \rda{}_\RR t_i$ is obtained by case (3) in the definition of
$\rda{}_\RR$ for all $1 \leqslant i \leqslant n$. If $n = 0$ then
$l\sigma = t$ and hence we can take $\tau = \sigma$ to satisfy
condition (a). So let $n > 0$. 
Let $p_i$ be the position of $s_i$ in $l\sigma$.
We distinguish two cases.
\begin{itemize}
\item
Suppose that $\seq{p} \notin \FPos(l)$.
We define a substitution $\tau$ as follows.
For $x \in \Var(l)$ let $q$ be the (unique) position in $\VPos(l)$ such
that $l|_q = x$. Let $P = \{ p_i \mid p_i \geqslant q \}$ be the set
of positions in $l\sigma$ of those terms $\seq{s}$ that occur in
$\sigma(x)$. We define $\tau(x)$ as the term that is obtained from
$\sigma(x)$ by replacing for all $p_i \in P$ the subterm $s_i$ at position
$p_i \backslash q$ with $t_i$. We have $t = l\tau$ and
$\sigma \rda{}_\RR \tau$, so condition (a) is satisfied.
\item
In the remaining case at least one position among $\seq{p}$ belongs to
$\FPos(l)$. Without loss of generality we assume that $p_1 \in \FPos(l)$.
Since $s_1 \rda{}_\RR t_1$ is obtained by case (3), $s_1 = l_1\mu$ and
$t_1 = r_1\nu$ for some rewrite rule $l_1 \to r_1$ and substitutions
$\mu$ and $\nu$ with $\mu \rda{}_\RR \nu$. We assume that $l_1 \to r_1$
and $l \to r$ share no variables.
Hence we may assume that $\mu = \sigma$.
We distinguish two further cases.
\smallskip
\begin{itemize}
\item
If $l_1 \to r_1$ and $l \to r$ are variants and $p_1 = \epsilon$ then
$n = 1$, $C = \Box$, and 
$l\sigma = s_1 = l_1\sigma \rda{}_\RR r_1\nu = t$.
Because
$l_1 \to r_1$ and $l \to r$ are variants, there exists a substitution
$\tau$ such that $r\tau = r_1\nu$ and $\sigma \rda{}_\RR \tau$. So in this
case condition (a) is satisfied.
\item
If $l_1 \to r_1$ and $l \to r$ are not variants or $p_1 \neq \epsilon$
then there exists an overlap $(l_1 \to r_1, p_1, l \to r)_\theta$ such
that $l\sigma = l\sigma[l_1\sigma]_{p_1}$ is an instance of
$l\theta = l\theta[l_1\theta]_{p_1}$. The TRS $\PR(\RR)$ contains the
rules $l\theta \to l\theta[r_1\theta]_{p_1}$ and
$l\theta \to r\theta$. The latter rule is used to obtain
$l\sigma \to_{\PR(\RR)} r\sigma$.
An application of the former rule yields
$l\sigma \to_{\PR(\RR)} l\sigma[r_1\sigma]_{p_1}$. From
$\sigma \rda{}_\RR \nu$ we infer that $r_1\sigma \rda{}_\RR r_1\nu = t_1$.
Hence $l\sigma \to_{\PR(\RR)} l\sigma[r_1\sigma]_{p_1} \rda{}_\RR
l\sigma[t_1]_{p_1} = C[t_1,s_2,\dots,s_n] \rda{}_\RR^* C[t_1,\dots,t_n] =
t$. The $\rda{}_\RR$-steps can be combined into a single one and hence
condition (b) is satisfied.
\qed
\end{itemize}
\end{itemize}
\end{proof}

The following example shows the necessity of left-linearity in the
preceding lemma.

\begin{example}
Consider the TRS $\RR$ consisting of the rewrite rules
\begin{xalignat*}{3}
\m{f}(x,x) &\to \m{a} &
\m{g}(x) &\to x &
\m{a} &\to \m{b}
\end{xalignat*}
Let $l = \m{f}(x,x)$, $r = \m{a}$, $\sigma(x) = \m{g}(\m{a})$, and
$t = \m{f}(\m{a},\m{g}(\m{b}))$. We have $l\sigma \rda{}_\RR t$ but
$t$ satisfies neither condition in Lemma~\ref{development}.
\end{example}

The final preliminary lemma states some obvious closure properties.

\begin{lemma}
\label{closure}
\mbox{}
\begin{enumerate}
\item
If $t \rda{\dd}^*_s u$ then $C[t] \rda{\dd}^*_{C[s]} C[u]$.
\smallskip
\item
If $t \rda{\ddg}^=_s u$ then $C[t] \rda{\ddg}^=_{C[s]} C[u]$.
\smallskip
\item
Let ${\geqslant} = {\to_\RR^*}$ and 
${\geqslant} \cdot {>} \cdot {\geqslant} \subseteq {>}$.
If $s > t$ and $t \rda{}^* u$ then $t \rda{\dd}^*_s u$.
\end{enumerate}
\end{lemma}
\begin{proof}
Straightforward.
\qed
\end{proof}

After these preliminaries we are ready for the main result.

\begin{theorem}
\label{dd-l2}
A left-linear TRS $\RR$ is confluent if 
${\cp} \subseteq {\downarrow}$ and $\PR(\RR)/\RR$ is terminating.
\end{theorem}
\begin{proof}
Because of Lemma~\ref{inclusion}, it is sufficient to prove 
confluence of $\rda{}_\RR$.
We show that the relation $\rda{}_\RR$ is extended locally decreasing
with respect to the source labeling. Labels are compared with
respect to the strict order ${>} = {\to_{\PR(\RR)/\RR}^+}$ and
the quasi-order ${\geqslant} = {\to_\RR^*}$. We show that
\[
{\mathrel{\test{\lda{}}{}{s}} \cdot \rda{}^{}_s}
\:\subseteq\:
{\rda{\ddg}^=_s
\cdot \rda{\dd}^*_s
\cdot \mathrel{\test{\lda{\dd}}{*}{s}}
\cdot \mathrel{\test{\lda{\ddg}}{=}{s}}}
\]
for all terms $s$ by well-founded induction on the order
$({>} \cup {\rhd})^+$.
In the base case $s$ is a variable and the inclusion trivially holds.
Let $s = f(\seq{s})$.
Suppose $t \lda{} s \rda{} u$.
We distinguish the following cases, depending on the derivation of
$s \rda{} t$ and $s \rda{} u$.
\begin{itemize}
\item
Neither $s \rda{} t$ nor $s \rda{} u$ is obtained by $(1)$, because 
$s$ is not a variable.
Suppose both $s \rda{} t$ and $s \rda{} u$ are obtained by $(2)$.
Then $t$ and $u$ can be written as $f(\seq{t})$ and $f(\seq{u})$.
Fix $i \in \{ 1, \dots, n \}$. We have $t_i \lda{} s_i \rda{} u_i$.
By the induction hypothesis there exist $t_i'$, $u_i'$, and $v_i$ such
that
\[
t_i \rda{\ddg}^=_{s_i} t_i' \rda{\dd}^*_{s_i}
v_i \mathrel{\test{\lda{\dd}}{*}{s_i}}
u_i' \mathrel{\test{\lda{\ddg}}{=}{s_i}}
u_i
\]
With repeated applications of Lemma~\ref{closure}(1,2) we obtain
\[
t
\rda{\ddg}^=_s f(\seq{t'}) \rda{\dd}_s^*
f(\seq{v}) \mathrel{\test{\lda{\dd}}{*}{s}}
f(\seq{u'}) \mathrel{\test{\lda{\ddg}}{=}{s}}
u
\]
\item
Suppose $s \rda{} t$ or $s \rda{} u$ is obtained by $(3)$.
Without loss of generality we assume that
$s \rda{} t$ is obtained by $(3)$, i.e.,
$s = l\sigma$, $t = r\tau$, and $\sigma \rda{} \tau$.
Following Lemma~\ref{development}, we distinguish the following two cases
for $l\sigma \rda{} u$.
\begin{itemize}
\item
Suppose $u \in \{ l\mu, r\mu \}$ for some $\mu$ with $\sigma \rda{} \mu$.
Fix $x \in \Var(l)$. We have $x\tau \lda{} x\sigma \rda{} x\mu$.
By the induction hypothesis there exist terms 
$t_x$, $u_x$, and $v_x$
such that
\[
x\tau \rda{\ddg}^=_{x\sigma}
t_x \rda{\dd}_{x\sigma}^*
v_x \mathrel{\test{\lda{\dd}}{*}{x\sigma}}
u_x \mathrel{\test{\lda{\ddg}}{=}{x\sigma}}
x\mu
\]
Define substitutions $\tau'$, $\nu$, and $\mu'$ as follows:
$\tau'(x) = t_x$,
$\nu(x) = v_x$, and
$\mu'(x) = u_x$ for all $x \in \Var(l)$, and
$\tau'(x) = \nu(x) = \mu'(x) = x$ for all $x \notin \Var(l)$.
We obtain
\[
t \rda{\ddg}^=_s
r\tau' \rda{\dd}_s^*
r\nu \mathrel{\test{\lda{\dd}}{*}{s}}
r\mu' \mathrel{\test{\lda{\ddg}}{=}{s}}
u
\]
by repeated applications of Lemma~\ref{closure}(1,2).
\item
In the remaining case we have
$s \to_{\PR(\RR)} u' \rda{} u$ for some term $u'$ as well as
$s \to_{\PR(\RR)} r\sigma$. Clearly $r\sigma \rda{} r\tau = t$.
Since $\RR$ is locally confluent (due to
${\cp} \subseteq {\downarrow}$),
there exists a term $v'$ such that
$r\sigma \to^* v' \mathrel{\test{\from}{*}{}} u'$
and thus also $r\sigma \rda{}^* v' \mathrel{\test{\lda{}}{*}{}} u'$.
We have $s > r\sigma$ and $s > u'$.
Lemma~\ref{closure}(3) ensures that
$r\sigma \rda{\dd}^*_s v' \mathrel{\test{\lda{\dd}}{*}{s}} u'$.
For every term $v$ with $s > v$ we have
\[
{\mathrel{\test{\lda{}}{}{v}} \cdot \rda{}^{}_v}
\:\subseteq\: 
{\rda{\ddg}^=_v
\cdot \rda{\dd}^*_v
\cdot \mathrel{\test{\lda{\dd}}{*}{v}}
\cdot \mathrel{\test{\lda{\ddg}}{=}{v}}}
\]
by the induction hypothesis. Hence the ARS
$\langle \TERM, \{ \rda{}_v \}_{v < s} \rangle$ is locally
decreasing and therefore the relation
\[
\rda{\dd}^{}_s ~=~ \bigcup_{v < s} \rda{}_v
\]
is confluent. This is used to obtain the diagram
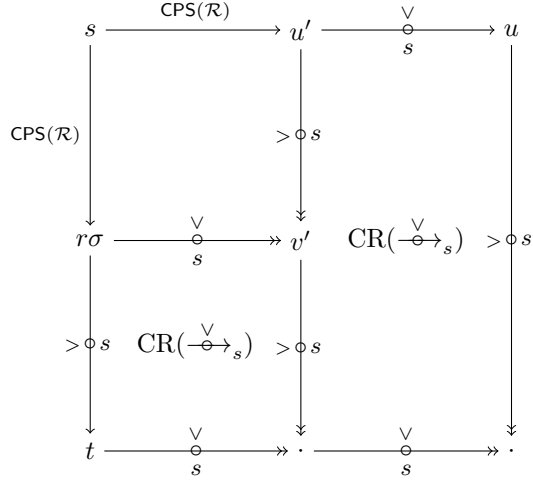
\[
\tikzstyle{element}=[rectangle]
\tikzstyle{ll}=[text width=0pt,text height=0pt]
\begin{tikzpicture}[scale=1.4]
\node at (0,4) [element] (s)  {$s$};
\node at (0,2) [element] (t') {$r\sigma$};
\node at (0,0) [element] (t)  {$t$};
\node at (2,0) [element] (w)  {$\cdot$};
\node at (2,4) [element] (u') {$u'$};
\node at (4,4) [element] (u)  {$u$};
\node at (2,2) [element] (v') {$v'$};
\node at (4,0) [element] (v)  {$\cdot$};
\draw [->] (s)  to (t');
\draw [->] (t') to (t) 
  node[pos=0.5] {$\circ$}
  node[pos=0.5,anchor=east] {\raisebox{2mm}{$\dD{180}$}}
  node[ll,pos=0.5,anchor=west] {\footnotesize $s$};
\draw [->] (s)  to (u');
\draw [->] (u') to (u) 
  node[pos=0.5] {$\circ$} 
  node[pos=0.5,anchor=south] {\raisebox{.5mm}{$\dd$}}
  node[pos=0.5,anchor=north] {\raisebox{-2mm}{\footnotesize $s$}};
\draw [->>] (t') to (v') 
  node[pos=0.5] {$\circ$}
  node[pos=0.5,anchor=south] {\raisebox{.5mm}{$\dd$}}
  node[pos=0.5,anchor=north] {\raisebox{-2mm}{\footnotesize $s$}};
\draw [->>] (v') to (w)
  node[pos=0.5] {$\circ$}
  node[pos=0.5,anchor=east] {\raisebox{2mm}{$\dD{180}$}}
  node[ll,pos=0.5,anchor=west] {\footnotesize $s$};
\draw [->>] (u') to (v') 
  node[pos=0.5] {$\circ$}
  node[pos=0.5,anchor=east] {\raisebox{2mm}{$\dD{180}$}}
  node[ll,pos=0.5,anchor=west] {\footnotesize $s$};
\draw [->>] (t)  to (w)
  node[pos=0.5] {$\circ$}
  node[pos=0.5,anchor=south] {\raisebox{.5mm}{$\dd$}}
  node[pos=0.5,anchor=north] {\raisebox{-2mm}{\footnotesize $s$}};
\draw [->>] (w)  to (v)
  node[pos=0.5] {$\circ$}
  node[pos=0.5,anchor=south] {\raisebox{.5mm}{$\dd$}}
  node[pos=0.5,anchor=north] {\raisebox{-2mm}{\footnotesize $s$}};
\draw [->>] (u)  to (v)
  node[pos=0.5] {$\circ$}
  node[pos=0.5,anchor=east] {\raisebox{2mm}{$\dD{180}$}}
  node[ll,pos=0.5,anchor=west] {\footnotesize $s$};
\node[anchor=east] at (0,3) {$\scriptstyle \PR(\RR)$};
\node[anchor=north] at (1,4.35) {$\scriptstyle \PR(\RR)$};
\node at (3,2) {CR($\rda{\dd}^{}_s$)};
\node at (1,1) {CR($\rda{\dd}^{}_s$)};
\end{tikzpicture}
\]
from which we conclude that
$t \rda{\dd}^*_s \cdot \mathrel{\test{\lda{\dd}}{*}{s}} u$.
\end{itemize}
\end{itemize}
\qed
\end{proof}

\begin{example}
\label{main-example}
Suppose we extend the TRS $\RR$ of Example~\ref{GL} by the rewrite rule
\[
\m{d}(x : y) \to  x : (x : \m{d}(y)) 
\]
The resulting TRS $\RR'$ has the same critical pair as $\RR$ and
$\PR(\RR')$ consists of
\begin{xalignat*}{2}
\m{inc}(\m{tl}(\m{nat})) &\to \m{tl}(\m{inc}(\m{nat})) &
\m{inc}(\m{tl}(\m{nat})) &\to \m{inc}(\m{tl}(\m{0} : \m{inc}(\m{nat}))) 
\end{xalignat*}
By taking the matrix interpretation (\cite{EWZ08})
\begin{xalignat*}{3}
\m{inc}_\MM(\vec{x}) & =
\begin{pmatrix} 1 & 0 \\ 1 & 0 \end{pmatrix} \vec{x} 
&
\m{hd}_\MM(\vec{x}) & = \vec{x}
&
\m{0}_\MM & = \begin{pmatrix} 0 \\ 0 \end{pmatrix}
\\
\m{nat}_\MM & =
\begin{pmatrix} 0 \\ 1 \end{pmatrix} 
&
\m{tl}_\MM(\vec{x}) & =
\begin{pmatrix} 1 & 1 \\ 1 & 0 \end{pmatrix} \vec{x} 
&
\m{s}_\MM(\vec{x}) & =
\begin{pmatrix} 1 & 1 \\ 0 & 0 \end{pmatrix} \vec{x}
\\
\m{d}_\MM(\vec{x}) & = 
\begin{pmatrix} 1 & 1 \\ 1 & 1 \end{pmatrix} \vec{x}
&
\m{:}_\MM(\vec{x},\vec{y}) & =
\begin{pmatrix} 1 & 1 \\ 1 & 1 \end{pmatrix} \vec{x} + \vec{y}
\end{xalignat*}
we obtain ${\RR' \subseteq {\geqslant_\MM}}$ and
$\PR(\RR') \subseteq {>_\MM}$:
\[
[\m{inc}(\m{tl}(\m{nat}))]_\MM =
\begin{pmatrix}
1 \\ 1
\end{pmatrix}
>
\begin{pmatrix}
0 \\ 0
\end{pmatrix}
= [\m{tl}(\m{inc}(\m{nat}))]_\MM
= [\m{inc}(\m{tl}(\m{0} : \m{inc}(\m{nat})))]_\MM
\]
Hence $\PR(\RR')/\RR'$ is terminating and Theorem~\ref{dd-l2} yields
the confluence of $\RR'$. Note that Theorem~\ref{dd-l1} does not apply
because $\PR(\RR') \cup \RR'_\du$ is not relatively terminating
with respect to $\RR_\nd$; consider the term $\m{d}(\m{nat})$.
\end{example}

The next example explains why one cannot replace $\PR(\RR)$ by
\[
\PR'(\RR) = \{ l_2\mu \to r_2\mu \mid
\text{$(l_1 \to r_1, p, l_2 \to r_2)_\mu$ is an overlap of $\RR$} \}
\]

\begin{example}
\label{counterexample2}
Consider the non-confluent TRS
\begin{xalignat*}{4}
\m{f}(\m{a}) &\to \m{c} &
\m{f}(\m{b}) &\to \m{d} &
\m{a} &\to \m{b} &
\m{b} &\to \m{a}
\end{xalignat*}
$\PR'(\RR)$ consists of
\begin{xalignat*}{2}
\m{f}(\m{a}) &\to \m{c} &
\m{f}(\m{b}) &\to \m{d}
\end{xalignat*}
and it is easy to see that $\PR'(\RR)/\RR$ is terminating.
\end{example}

An easy extension of our main result is obtained by excluding critical
pair steps from $\PR(\RR)$ that give rise to trivial critical pair steps.
The proof is based on the observation that Lemma~\ref{development} still
holds for this modification of $\PR(\RR)$.

\section{Automation}
\label{automation}

Concerning the automation of Theorem~\ref{dd-l2}, for checking relative
termination we use the following criteria of Geser~\cite{G90}:

\begin{theorem}
For TRSs $\RR$ and $\SS$, $\RR/\SS$ is terminating if
\begin{enumerate}
\item
$\RR = \varnothing$,
\item
$\RR \cup \SS$ is terminating, or
\item
there exist a well-founded order $>$ and a quasi-order $\geqslant$ such
that
$>$ and $\geqslant$ are closed under contexts and substitutions,
${\geqslant} \cdot {>} \cdot {\geqslant} \subseteq {>}$,
$\RR \cup \SS \subseteq {\geqslant}$, and
$(\RR \setminus {>})/(\SS \setminus {>})$ is terminating.
\end{enumerate}
\end{theorem}

Based on this result, termination of $\PR(\RR)/\RR$ is shown by repeatedly
using the last condition to simplify $\PR(\RR)$ and $\RR$. As soon as the
first condition applies, termination is concluded. If the first condition
does not apply and the third condition does not make progress, we try to
establish termination of $\PR(\RR) \cup \RR$ using the termination tool
\TTTT~\cite{TTT2}. For checking the third condition we use matrix
interpretations~\cite{EWZ08}. 

In the remainder of this section we show how to implement
Theorem~\ref{dd-ll}. We start by observing that the condition of
Theorem~\ref{dd-ll} is undecidable even for locally confluent 
TRSs.\footnote{Contradicting the claim in \cite[Section~4.2]{vO08}.}

\begin{lemma}
The following decision problem is undecidable:
\begin{center}
\begin{tabular}{l@{\quad}p{60ex}}
instance: & a finite locally confluent linear TRS $\RR$, \\
question: & are all critical pairs locally decreasing with respect to
the rule-labeling heuristic? \\
\end{tabular}
\end{center}
\end{lemma}
\begin{proof}
We reduce the problem to the joinability of two ground terms for linear
non-overlapping TRSs. The latter problem is undecidable as an easy
consequence of the encoding of Turing machines as linear non-overlapping
TRSs, see e.g.\ \cite{TeReSe}. So let $\SS$ be a (finite) linear
non-overlapping TRS and let $s$ and $t$ be arbitrary ground terms. We
extend $\SS$ with fresh constants $\m{a},\m{b}$ and the rewrite rules 
\begin{xalignat*}{4}
\m{a} & \to s & 
\m{b} & \to t &
\m{a} & \to \m{b} & 
\m{b} & \to \m{a} 
\end{xalignat*}
to obtain the TRS $\RR$. If $s$ and $t$ are joinable (in $\SS$) then 
the four critical pairs are locally decreasing by labeling
the above four rules with $1$ and all rules in $\SS$ with $2$ and
using the order $1 > 2$. If $s$ and $t$ are not joinable,
then no rule-labeling will make the critical pairs locally
decreasing. So confluence of $\RR$ can be established by the rule-labeling
heuristic if and only if the terms $s$ and $t$ are joinable in $\SS$.
\qed
\end{proof}

By putting a bound on the number of steps to check joinability
we obtain a decidable condition for (extended) local decreasingness:
\[
l_2[r_1]_p\mu 
\:
\mathrel{\overbrace
{\rule[3ex]{0pt}{0pt}\smash{\xrightarrow{\dd}}_\alpha^*
\cdot \xrightarrow{\ddg}_\beta^{=}
\cdot \mathrel{\smash{\xrightarrow{\dd}}_{\alpha\beta}^*}
}^{\text{at most $k$ steps}}}
\cdot 
\mathrel{\overbrace
{\rule[3ex]{0pt}{0pt}\mathrel{\test{\xleftarrow{\dd}}{*}{\alpha\beta}}
\cdot \mathrel{\test{\xleftarrow{\ddg}}{=}{\alpha}}
\cdot \mathrel{\test{\xleftarrow{\dd}}{*}{\beta}}
}^{\text{at most $k$ steps}}}
\:
r_2\mu
\]
for each overlap $(l_1 \to r_1,p,l_2 \to r_2 )_\mu$ of $\RR$
with $\alpha = l_1 \to r_1$ and $\beta = l_2 \to r_2$.
Below we reduce this to \emph{precedence constraints} of the form
\[
\phi ::= 
\top
\mid \bot
\mid \phi \lor \phi
\mid \phi \land \phi
\mid \alpha > \alpha
\mid \alpha \geqslant \alpha
\]
where $\alpha$ stands for variables corresponding to the rules in $\RR$.
From the encodings of termination methods for term rewriting, we know
that the satisfiability of such precedence constraints is easily
determined by SAT or SMT solvers (cf.\ \cite{CLS06,ZHM09}).

\begin{definition}
For terms $s$, $t$ and $k \geqslant 0$,
a pair
\(
((\gamma_1, \dots, \gamma_m), (\delta_1, \dots, \delta_n))
\)
is called a \emph{$k$-join instance} of $(s,t)$ if $m, n \leqslant k$, 
$\gamma_1, \dots, \gamma_m, \delta_1, \dots, \delta_n \in \RR$, and
\[
s \to_{\gamma_1}^{} \cdot \: \cdots \: \cdot \to_{\gamma_m}^{} \cdot
\From{\delta_n}{} \cdot \: \cdots \: \cdot \From{\delta_1}{} t
\]
The \emph{embedding order} $\sqsupseteq$ on sequences is defined as
$(a_1,\ldots,a_n) \sqsupseteq (a_{i_1},\ldots,a_{i_m})$ whenever
$1 \leqslant i_1 < \cdots < i_m \leqslant n$.
The set of all minimal (with respect to
${\sqsupseteq} \times {\sqsupseteq}$) $k$-join instances of $(s,t)$ is
denoted by $J_k(s,t)$.
We define $\Phi^\alpha_\beta((\gamma_1,\ldots,\gamma_n))$ as
\[
\bigvee_{i \leqslant n}
\biggl( \bigwedge_{j < i} \alpha > \gamma_j
\land \Psi_{i,n}
\biggr) 
\]
with $\Psi_{i,n}$ denoting $\top$ if $i = n$ and
\[
\beta \geqslant \gamma_i \,\land\, \bigwedge_{i < j \leqslant n}
(\alpha > \gamma_j \,\lor\, \beta > \gamma_j)
\]
if $i < n$. Furthermore, $\RL_k(\RR)$ denotes the conjunction of
\[
\bigvee \bigl\{ \Phi^{l_1 \to r_1}_{l_2 \to r_2}(\vec \gamma) \land
\Phi^{l_2 \to r_2}_{l_1 \to r_1}(\vec \delta) \bigm|
(\vec \gamma, \vec \delta) \in J_k(l_2[r_1]_p\mu,r_2\mu) \bigr\}
\]
for all overlaps $(l_1 \to r_1,p,l_2 \to r_2)_\mu$ of $\RR$.
\end{definition}

The only non-trivial part of the encoding is the minimality 
condition in $J_k(s,t)$. The next lemma explains why
non-minimal pairs can be excluded from the set and Example~\ref{rl}
shows the benefit of doing so.

\begin{lemma}
If $\Phi^\alpha_\beta(\vec\delta)$ is satisfiable and
$\vec\delta \sqsupseteq \vec\gamma$ then $\Phi^\alpha_\beta(\vec\gamma)$
is satisfiable.
\end{lemma}
\begin{proof}
Straightforward.
\qed
\end{proof}

We illustrate the encoding on a concrete example.

\begin{example}
\label{rl}
Consider the TRS $\RR$ of Example~\ref{GL}. We show how $\RL_4(\RR)$ is
computed. There is a single overlap $(1,11,5)_\epsilon$ resulting in the
critical pair $s \approx t$ with
$s = \m{inc}(\m{tl}(\m{0} : \m{inc}(\m{nat})))$ and
$t = \m{tl}(\m{inc}(\m{nat}))$.
Its $4$-join instances are
\begin{xalignat*}{4}
((3), (1, 4, 3)) 
&& ((3, 1), (1, 4, 3, 1)) 
&& ((3, 1), (1, 4, 1, 3))
&& ((3, 1), (1, 1, 4, 3)) 
\\
&& ((1, 3), (1, 4, 3, 1))
&& ((1, 3), (1, 4, 1, 3))
&& ((1, 3), (1, 1, 4, 3))
\end{xalignat*}
Only the first one belongs to $J_4(s,t)$ and hence
\[
\RL_4(\RR) \:=\: \Phi^1_5((3)) \land \Phi^5_1((1,4,3))
\]
with $\Phi^1_5((3)) ~=~ 5 \geqslant 3 \,\lor\, 1 > 3$ and
\begin{align*}
\Phi^5_1((1,4,3)) ~=  
     & ~ (1 \geqslant 1 \,\land\, (1 > 4 \,\lor\, 5 > 4) \,\land\,
         (1 > 3 \,\lor\, 5 > 3)) \\
\lor & ~ (5 > 1 \,\land\, 1 \geqslant 4 \,\land\,
         (1 > 3 \,\lor\, 5 > 3)) \\
\lor & ~ (5 > 1 \,\land\, 5 > 4 \,\land\, 1 \geqslant 3) \\
\lor & ~ (5 > 1 \,\land\, 5 > 4 \,\land\, 5 > 3)
\end{align*}
This formula is satisfied by taking (e.g.) the order $5 > 1, 3, 4$. Hence,
the confluence of $\RR$ is concluded by local decreasingness with respect
to the rule labeling heuristic using at most 3 steps to close critical
pairs.
\end{example}

\begin{theorem}
\label{dd-rl}
A linear TRS $\RR$ is confluent if $\RL_k(\RR)$ is satisfiable for some
$k \geqslant 0$.
\qed
\end{theorem}

\section{Experimental Results}
\label{experiments}

\begin{table}[t]
\begin{center}
\begin{tabular}{@{}l@{\qquad}rrrrrrr@{}}
\hline \\[-1ex]
& (a)
& (b)
& (c)
& (d)
& (e)
& (f)
& (g)
\\[1ex]
YES
& \phantom{1}20
& \phantom{1}81
& \phantom{1}67
& \phantom{1}49
& 100
& 107
& 135 
\\
timeout ($60$ s)
& 0
& 0
& 0
& 3
& 4
& 3
& 17
\\[1ex]
\hline
\end{tabular}
\end{center}
\caption{Summary of experimental results}
\label{summary}
\end{table}

We tested our methods on a collection of 424 TRSs, consisting
of the 103 TRSs in the ACP distribution,\footnote{%
\url{http://www.nue.riec.tohoku.ac.jp/tools/acp/}}
the TRSs of Examples~\ref{counterexample1}, \ref{main-example}, and
\ref{counterexample2}, and those TRSs in
version 5.0 of the Termination Problems Data Base\footnote{%
\url{http://termination-portal.org/wiki/TPDB}}
that are either non-terminating or not known to be
terminating.
(Systems that have extra variables in right-hand sides of rewrite rules
are excluded.)
The results are summarized in Table~\ref{summary}.
The following techniques are used to produce the columns:
\begin{itemize}
\item[(a)]
Knuth and Bendix criterion~\cite{KB70}: termination and joinability of all
critical pairs,
\item[(b)]
orthogonality,
\item[(c)]
Theorem~\ref{dd-rl} with $k = 4$,
\item[(d)]
Theorem~\ref{dd-l1},
\item[(e)]
Theorem~\ref{dd-l2},
\item[(f)]
The extension of Theorem~\ref{dd-l2} mentioned at the end of
Section~\ref{confluence via relative termination} in which critical pair
steps that generate trivial critical pairs are excluded from $\PR(\RR)$,
\item[(g)]
ACP~\cite{ACP}.
\end{itemize}
To obtain the data in columns (a)--(f) we slightly extended the
open source termination tool \TTTT. For the data in column (c) the
SAT solver \textsf{\mbox{Mini}Sat}~\cite{MINISAT} is used.
Since local confluence is undecidable (for non-terminating TRSs),
in (c)--(f) it is approximated by
${\cp} \subseteq \bigcup \{ {\to^i} \cdot {\From{}{j}} \mid
i, j \leqslant 4 \}$.

ACP proves that 198 of the 424 TRSs are not confluent.
Of the remaining 226 TRSs, local confluence can be shown using at most
4 rewrite steps from both terms in every critical pair for 185 TRSs.
Moreover, of these 185 TRSs, 147 are left-linear and 75 are linear.
As a final remark, the combination of (c) and (f) proves that
129 TRSs are confluent, and the combination of (c), (f), and (g)
shows confluence for 145 TRSs.
These numbers clearly show that both our results have a role to play.

\section{Conclusion}
\label{conclusion}

In this paper we presented two results based on the decreasing diagrams
technique for proving confluence of TRSs. For linear TRSs we showed how
the rule-labeling heuristic can be implemented by means of
an encoding as a satisfiability problem and we employed the self-labeling
heuristic to obtain the result that an arbitrary left-linear locally
confluent TRS is confluent if its critical pair steps are relatively
terminating with respect to its rewrite rules. We expect that both results
will increase the power of ACP~\cite{ACP}.

As future work we plan to investigate whether the latter result can be
strengthened by decreasing the set $\PR(\RR)$ of critical pair steps
that need to be relatively terminating with respect to $\RR$. 
We anticipate that some of the many critical pair criteria for confluence
that have been proposed in the literature (e.g.\ \cite{H80,vO97}) can be
used for this purpose. The idea here is to exclude the critical pair
steps that give rise to critical pairs whose joinability can be shown by
the conditions of the considered criterion.

\subsection*{Acknowledgements}

We thank Mizuhito Ogawa, Vincent van Oostrom, and Harald Zankl for
scrutinizing an earlier version of this paper.

\end{document}